\documentclass[conference]{IEEEtran}
\usepackage{cite}

\ifCLASSINFOpdf
  \usepackage[pdftex]{graphicx}

\else

\fi

\usepackage[cmex10]{amsmath}
\usepackage{mdwmath}
\usepackage{amsbsy}
\usepackage{eqparbox}
\usepackage{multirow}

\usepackage{algorithm}							
\usepackage{algpseudocode}

\IEEEoverridecommandlockouts 	

\usepackage{amsmath}
\usepackage{amsfonts}
\usepackage{amssymb}
\usepackage{url}

\newtheorem{lemma}{Lemma}
\newtheorem{proposition}{Proposition}

\newtheorem{assumption}{Assumption}

\usepackage{enumerate}

\usepackage{array}

\usepackage[english]{babel}

\usepackage{graphicx} 
\usepackage{color}
\usepackage{lineno, blindtext}

\renewcommand{\baselinestretch}{0.98}

\def\realnumbers{\mathbb{R}}
\def\complexnumbers{\mathbb{C}}

\def \graph	{\mathcal{G}}		
\def \nodes	{\mathcal{V}}		
\def \adm {{\boldsymbol Y}}

\def \green {{\boldsymbol X}}
\def \edges	{\mathcal{E}}		

\def \topology {{\boldsymbol T}}

\def \nonodes {{n}}
\def \noswitches {{r}}
\def \nosensors {{p}}
\def \noedges {{w}}

\def\1{{\mathbf{1}}}
\begin{document}

\title{Data-Driven Approach for Distribution Network Topology Detection}

\author{
\IEEEauthorblockN{G. Cavraro}
\IEEEauthorblockA{DEI - University of Padova\\
Padova, Italy\\
cavraro@dei.unipd.it
}
\and
\IEEEauthorblockN{R. Arghandeh}
\IEEEauthorblockA{EECS - U. C. Berkeley\\
Berkeley, USA\\
arghandeh@berkeley.edu
}
\and
\IEEEauthorblockN{K. Poolla}
\IEEEauthorblockA{EECS - U.C. Berkeley\\
Berkeley, USA\\
poolla@berkeley.edu
}
\and
\IEEEauthorblockN{A. von Meier}
\IEEEauthorblockA{EECS - U.C. Berkeley\\
Berkeley, USA\\
vonmeier@berkeley.edu
}

}

\maketitle

\begin{abstract}
This paper proposes a data-driven approach to detect the switching actions and  topology transitions in distribution networks. It is based on the real time analysis of time-series voltages measurements. The analysis approach draws on data from high-precision phasor measurement units ($\mu$PMUs or synchrophasors) for distribution networks. The key fact is that time-series measurement data taken from the distribution network has specific patterns representing state transitions such as topology changes. The proposed algorithm is based on comparison of actual voltage measurements with a library of signatures derived from the possible topologies simulation. The IEEE 33-bus model is used for the algorithm validation. 
\end{abstract}

\section{Introduction}
Different tools have been developed and implemented to monitor distribution network behavior with more detailed and temporal information, such as SCADA, smart meters and line sensors. Creating observability out of disjointed data streams still remains a challenge, though\cite{vonMeier2014Chap}.  The cost for monitoring systems in distribution networks still remains a barrier to equipping all nodes with measurement devices. To some extent, a capable distribution state estimator can compensate for the lack of measurement data to support system observability. However, topology errors will easily downgrade state estimator accuracy. Topology detection is a key component for different real-time operation and control functions.
%
%
Most of literature on topology detection is based on state estimator (SE) results and measurement matching with different topologies. In \cite{korres2012} authors propose a state estimation algorithm that incorporates switching device status as additional state variables. A normalized residual test is used to identify the best estimate of the topology. SE-based algorithms are easy to implement, but their accuracy is limited to that of the state estimator. They are also sensitive to measurement device placement. In \cite{sharon2012}, the authors provide a tool for choosing sensor placement for topology detection. Given a particular placement of sensors, the tool reveals the confidence level at which the status of switching devices can be detected. Authors in \cite{Ciobotaru2007} are focused on estimating the impedance at the feeder level. However, even a perfect identification of network impedance cannot always guarantee the correct topology, since multiple topologies could present very similar impedances.


In this paper,  a real time topology detection algorithm is proposed based on time series analysis of phasor measurement unit (PMU) data. This approach is inspired by high-precision phasor measurement units for distribution systems, called micro-synchrophasors or ($\mu$-PMU), with whose development the authors are involved \cite{vonMeier2014}. The main idea derives from the fact that time-series data from a dynamic system show specific patterns regarding system state transitions, a signature is left from each topology change. 
The algorithm is based on the comparison of the trend vector, built from system observations, with a library of signatures derived from the possible topology transitions. The topology detection results are impacted by load uncertainty and measurement device accuracy. Therefore, the analysis takes load dynamics and measurement error into account. 
The topology detection accuracy is also depends on the number of $\mu$-PMUs. But, the simulations shows topology detection is converge robustly even with limited measurement devices. 


\section{Distribution Network Model and Physical Topology}
\label{sec:model}
Given a matrix $W$, we denote its element-wise complex conjugate by $\overline{W}$ , its transpose by $W^T$  and its conjugate transpose by $W^*$. We denote the matrices of the absolute value, of the real and of the imaginary part of $W$ by $|W|$, $\Re(W)$ and by $\Im(W)$ , respectively, and with $|W|$ . We denote the entry of $W$ that belongs to the $j$-th row and to the $k$-th column by $[W]_{jk}$. 
Given a vector $v$, $[v]_j$ will denote its $j$-th entry, while $[v]_{-j}$ the subvector of $v$, in which the $j$-th entry has been eliminated. Given two vectors $v$ and $w$, we denote by $\langle v,w \rangle$ their inner product $v^*w$. 
We define the column vector of all ones by $\1$
We associate with the electric grid the directed graph $\graph= (\nodes, \edges)$, where $\nodes$ is the set of nodes (the buses), with cardinality $\nonodes$ and $\edges$ is the set of edges (the electrical lines connecting them), with cardinality $\noedges$; the set $\mathcal S$ of the switched deployed in the electrical grid, with cardinality $\noswitches$ and the set $\mathcal P$ is the set of the electrical grid nodes endowed with voltage phasor measurement units (PMUs), with cardinality $\nosensors$.
Let $A \in \{0, \pm 1\}^{\noedges \times \nonodes}$ 
be the incidence matrix of the graph $\graph$, $A=
\begin{bmatrix}
a_1^T&\dots&
a_\noedges^T
\end{bmatrix}^T
$
where $a_j$ is the $j$-th row of $A$, whose elements are all zeroes except for the entries associated to the nodes connected by the $j$-th edge, for which the elements equal $+1$ or $-1$,  respectively.  If the graph $\graph$ is connected (i.e. for every pair of nodes there is a path connecting them), then $\1$ is the only vector in the null space $\ker A$, $\1$ being the column vector of all ones.
In this study, we limit our study to the steady state behavior of the system, when all voltages and currents are sinusoidal signals waving at the same frequency $\omega_0$. Thus, they can be expressed via a complex number whose magnitude corresponds to the signal root-mean-square value, and whose phase corresponds to the phase of the signal with respect to an arbitrary global reference. Therefore, $x$ represents the signal
$
x(t) = |x| \sqrt{2} \sin(\omega_0 t + \angle x).
$

We will denote the vector of the voltages as $u \in \complexnumbers^\nonodes$, the vector of the currents as $i \in \complexnumbers^\nonodes$, and the vectors of the powers as $s = p + i q \in \complexnumbers^\nonodes$, with $p,q \in \realnumbers^\nonodes$  are the active and the reactive power injected at node $v$. The state of the switches is 
 $\sigma \in [0,1]^ \noswitches$, where $[\sigma]_l = 0$ if the switch $v$ is open, $[\sigma]_l = 1$ if the switch $l$ is closed. The measured grid voltages are collected in $y \in \complexnumbers^\nosensors$. We define the \emph{trend vector} $\delta(t_1,t_2)\in \complexnumbers^\nosensors$, as the difference between phasorial voltages taken at the two time instants $t_1$ and $t_2$. i.e. $ \delta(t_1,t_2)={u}(t_1)-{u}(t_2)$.
We assume that the deployed PMUs in the distribution network take measurements at the frequency $f$. 

We consider a topology $\topology_\sigma$ which switches status are described by $\sigma$. Its bus admittance matrix $\adm_\sigma$ is defined as
\begin{equation}
[\adm_\sigma]_{jk} = \begin{cases}
\sum_{j \neq k} Y_{jk}, \text{ if } j = k \\
- Y_{jk}, \text{ otherwise }
\end{cases}
\label{eq:busAdmMatrixDef}
\end{equation}
where $Y_{jk}$ is admittance of the branch connecting bus $j$ and bus $k$, we neglect the shunt admittances. From \eqref{eq:busAdmMatrixDef} we see that $\adm_\sigma$ is symmetric and it satisfies
\begin{equation}
\adm_\sigma \1 = 0,
\label{eq:admittance_null_space}
\end{equation}
i.e. $\1$ belongs to the Kernel of $\adm_\sigma$. Furthermore, it can be shown that if $\graph$, the graph associated to the electrical grid, is connected, then the kernel of $\adm_\sigma$ has dimension 1.

We model the substation as an ideal sinusoidal voltage source (\emph{slack bus}) at the distribution network nominal voltage $U_N$, with arbitrary and fixed angle $\phi$. We  consider, without loss of generality, $\phi=0$. We model all nodes except the substation as \emph{constant power devices}, or \emph{P-Q buses}. 
The system state satisfies the following equations
\begin{align}
&i = \adm_\sigma u \label{eq:nodevoltage}\\
&u_0 = U_N  \label{eq:PCCidealvoltgen}\\
&u_v i_v^* = p_v + i q_v \qquad v\neq 0 \label{eq:nodeconstpwr} 
\end{align}

The following Lemma \cite{Zampieri} introduces a particular and useful pseudo inverse of $\adm_\sigma$ for our topology detection algorithm.
\begin{lemma}
There exists a unique symmetric, positive semidefinite matrix $\green_\sigma \in \complexnumbers^{n\times n}$ such that
\begin{equation}
\begin{cases}
\green_\sigma \adm_\sigma = I - \1 \1_0^T \\
\green_\sigma \1_0 = 0.
\end{cases}
\label{eq:Xproperties}
\end{equation}
\label{lemma:X}
\end{lemma}
%
%
Applying Lemma \ref{lemma:X}, from \eqref{eq:nodevoltage} and \eqref{eq:PCCidealvoltgen} we can express voltages of the grid as a function of the currents and of the nominal voltage
\begin{equation}
u=\green_\sigma i+\1 U_N
\label{eq:u=Xi}
\end{equation}
The following proposition (\cite{Zampieri}) provides a approximation of the relationship between voltages and powers.
\begin{proposition}
Consider the physical model described by the set of nonlinear equations \eqref{eq:nodevoltage}, \eqref{eq:PCCidealvoltgen},  \eqref{eq:nodeconstpwr} and \eqref{eq:u=Xi}.
Node voltages then satisfy
\begin{equation}
u
=
U_N \1 + 
\frac{1}{U_N}
\green_\sigma
\bar s
+
o\left(\frac{1}{U_N}\right)
\label{eq:approximate_solution}
\end{equation}
(the little-o notation means that $\lim_{U_N\rightarrow \infty} \frac{o(f(U_N))}{f(U_N)} = 0$).
\label{pro:approximation}
\end{proposition}
Equation \eqref{eq:approximate_solution} is derived from a first order Taylor expansion w.r.t. the nominal voltage $U_N$ of the equation relates powers and voltages. The approximate solution of \eqref{eq:approximate_solution} has been already used with success in state estimation \cite{schenato2014bayesian},  Volt/Var optimization \cite{cavraroTAC2015}, and the optimal power flow problem \cite{cavraro2014cdc}. 

There is always some noise associated with PMUs, i.e. the output of our PMU placed at bus $j$ is 
\begin{small}
\begin{equation}
y_j = u_j + e_j
\label{eq:PMUmeasurements}
\end{equation}
\end{small}
where $e_j \in \complexnumbers$ is the error caused by the measurement device. A common index for measurement error is the \emph{total vector error} (TVE) \cite{6111219}.
%
In this paper we assume that the loads have constant \emph{power factor}, and consequently 
\begin{equation}
\frac{[p(t_1)]_j}{[q(t_1)]_j} = \frac{[p(t_2)]_j}{[q(t_2)]_j}, \quad \forall j, t_1, t_2.
\end{equation}
Furthermore, loads have dynamic behavior, described by
\begin{align}
p(t+1) &= p(t) + n_p(t) 
\label{eq:deltap}
\end{align}
where $n_p(t)$ is a Gaussian random variable, $n_p(t) \sim \mathcal N(0,\sigma^2_p p_M^2)$.  A load measurement data set for five residential houses in the Texas, U. S. has been analyzed to drive the statistical load model. Some smart meters can measure loads every second or a couples of seconds. Load demand (kW) are recorded every seconds for a week. Statistical analysis of load variations between two consecutive seconds is presented in Table~\ref{tab:loadDataStat2}. 
\begin{table}
\caption{Load differences}
\label{tab:loadDataStat2}
\centering
\begin{tabular}{lcc}
           & Mean (kW) & SD (kW) \\
House 1    & 0.000     & 0.045   \\
House 2    & 0.000     & 0.070   \\
House 3    & 0.000     & 0.113   \\
House 4    & 0.000     & 0.110   \\
House 5    & 0.000     & 0.046   \\
Aggregate  & 0.000     & 0.184   
\end{tabular}
\end{table}
\begin{table}
\caption{Aggregate load differences for differents frequency}
\label{tab:loadDataStat3}
\centering
\begin{tabular}{lccc}
           & Mean (kW) & SD (kW) \\
$f=1\text{ Hz}$  	& 0.000     & 0.184   \\
$f=0.2\text{ Hz}$	& 0.000     & 0.425   \\
$f=0.1\text{ Hz}$	& 0.000     & 0.604   
\end{tabular}
\end{table}%
In the United States, a number of houses are connected to one distribution transformer. Therefore, the aggregated loads for five houses are considered as the reference for load variability in this paper. Lower measurement sampling time leads to higher uncertainty in load data variability. In Table \ref{tab:loadDataStat3}, the aggregate characterization for different frequencies is reported.

\section{Identification of Switching Actions }
\label{sec:propagation_action_switches}
The basic idea behind our proposed approach is that changes in switching status will create specific signatures in the voltage waveform measurements.
In order to develop the theoretical base for the proposed algorithm and its ease of mathematical proof, we make the following assumptions. 
%
%
\begin{assumption}
\label{ass:sameR/X}
All the lines have the same resistance over reactance ratio. Therefore, $\Im(Y_{jk}) = \alpha \Re(Y_{jk}), \forall Y_{jk}$.
\end{assumption}
\begin{assumption}
\label{ass:statuschng}
Only one switch can change its status at each time.
\end{assumption}
\begin{assumption}
\label{ass:connection}
The graph associated to the electrical network is always connected, i.e. that there are no admissible state in which any portion of the grid remains disconnected.
\end{assumption}
\begin{assumption}
\label{ass:known}
The initial switches status are known.
\end{assumption}
%
Assumption \ref{ass:sameR/X} will be relaxed in Section \ref{sec:results}, in order to test the algorithm in a more realistic scenario. However, it allows us to decompose the bus admittance matrix as follow.
\begin{equation}
\adm_{\sigma} = U \Sigma_R U^* + i U \Sigma_I U^*
\end{equation}
where $\Sigma_R,\Sigma_I$ are diagonal matrices whose diagonal entries are the non-zero eigenvalues of $\Im(\adm_{\sigma(t-1)})$ and $\Re(\adm_{\sigma(t-1)})$, $U$ is an orthonormal matrix that includes all the associated eigenvectors and $\Sigma_I = \alpha \Sigma_R$. From \eqref{eq:admittance_null_space}, it can be showed that $U$ spans the space orthogonal to $\1$. 
Furthermore, we have
\begin{equation}
\green_{\sigma} = (1 + i\alpha)^{-1} \Gamma U (\Sigma_R)^{-1} U^* \Gamma 
\label{eq:X1}
\end{equation}
with $\Gamma = (I - \1 e_0^T)$. Assumption \ref{ass:statuschng} is reasonable for the proposed algorithm framework: it works on a time scale of seconds, and typically the switches are electro-mechanical devices and their actions are not simultaneous. Finally, Assumption \ref{ass:connection} is always satisfied during the normal operation.

Assume that at time $t-1$ the switches status is described by $\sigma(t-1) =  \sigma_1$, resulting in the topology $\topology_{\sigma(t-1)}$ with bus admittance matrix $\adm_{\sigma(t-1)}$. Applying Proposition \ref{pro:approximation} and neglecting the infinitesimal term, the voltages can be expressed as
\begin{equation}
{u}(t-1)=
\green_{\sigma(t-1)} 
\frac{\bar s }{U_N} + \1 U_N
\label{eq:ux1}
\end{equation}
At time $t$ the $\ell$-th switch, that was previously open, changes its status. Let the new status be described by $\sigma(t) =  \sigma_2$, associated to the topology is $\topology_{\sigma(t)}$. Since we are basically adding the edge in which switch $\ell$ is placed from the graph that represents the grid, we can write 
\begin{equation}
\adm_{\sigma(t)} = \adm_{\sigma(t-1)} + y_{\ell} a_{\ell} a_{\ell}^T
\end{equation}
where  $y_{\ell}$ is the admittance of the line, and $a_{\ell}$ is the $\ell$-th row of the adjacency matrix associated with the $\topology_{\sigma(t)}$. %
Since $ a_{\ell} $ is orthogonal to  $\1$, there exists $b_{\ell}$ such that $Ub_{\ell} = a_{\ell}$. This allow us to write
\begin{small}
\begin{align}
\adm_{\sigma(t)} &= (1 + i\alpha)  U (\Sigma_R + \Re(Y_{\ell}) b_{\ell} b_{\ell}^T) U^* \notag\\
\green_{\sigma(t)} &= (1 + i\alpha)^{-1} \Lambda U (\Sigma_R + \Re(Y_{\ell}) b_{\ell} b_{\ell}^T)^{-1} U^* \Lambda^T \label{eq:X2}
\end{align}
\end{small}
The voltages satisfy
\begin{equation}
{u}(t) =
\green_{\sigma(t)} \frac{\bar s }{U_N}
+ \1 U_N
\label{eq:ux2}
\end{equation}
From \eqref{eq:X1} and \eqref{eq:X2}, the trend vector can be written as
\begin{equation}
\delta(t,t-1) = \Gamma \Phi_{\sigma(t-1)\sigma(t)} \Gamma^T \frac{\bar s }{U_N}
\label{eq:Delta1}
\end{equation}
where 
\begin{equation}
\Phi_{\sigma(t-1)\sigma(t)} = U \Sigma_R^{-1} U^* - U (\Sigma_R + \Re(Y_{\ell}) b_{\ell} b_{\ell} ^T)^{-1} U^*
\end{equation}
We can observe that when there is a switching action, the voltage profile varies in accordance to a specific topology transition. Since $[\sigma_1]_{-\ell} = [\sigma_2]_{-\ell}$, for the ease of notation in the following we will write $\Phi_{\sigma(t-1)\sigma(t)}$ as $\Phi_{[\sigma(t)]_{-\ell}}$.
The following Proposition shows a characteristic of $\Phi_{[\sigma(t)]_{-\ell}}$ that is crucial for  the development of our topology detection algorithm.
\begin{proposition}
For every topology transition from the state described by $\sigma(t-1)$ to the one described by $\sigma(t)$ by changing the switch $\ell$, $\Phi_{[\sigma(t)]_{-\ell}}$ is a rank one matrix.
\label{pro:Phi_rank}
\end{proposition}
\begin{proof}
Exploiting \eqref{eq:X1}, \eqref{eq:X2}, using Ken Miller Lemma \cite{miller1981inverse} with some simple computations, we can write
\begin{equation}
\Phi_{[\sigma(t)]_{-\ell}} = 
\mu U \Sigma_R^{-1} b_\ell b_\ell^T \Sigma_R^{-1} U^*
\label{eq:PhiMiller}
\end{equation}
with 
$$\mu = \frac{1}{1 + \text{Tr}(\Re(Y_\ell) b_{\ell} b_{\ell}^T\Sigma_R)}.$$
It's trivial to see that $\Phi_{[\sigma(t)]_{-\ell}}$ is a rank one matrix with the non-zero eigenvalue
$\lambda_{\sigma(t)_{-\ell}} = \mu \|U \Sigma_R^{-1} b_{\ell} \|^2$
associated with the eigenvector 
$\hat g_{[\sigma(t)]_{-\ell}} = U \Sigma_R^{-1} b_\ell
$%
and thus can be written as 
$$\Phi_{[\sigma(t)]_{-\ell}} = \lambda_{[\sigma(t)]_{-\ell}} \hat g_{[\sigma(t)]_{-\ell}} \hat g_{[\sigma(t)]_{-\ell}}^*.$$
\end{proof}
The trend vector $\delta(t)$ shows the relationship between switching actions and voltage profile. Thanks to Proposition \ref{pro:Phi_rank} we can write it as
$$\delta(t,t-1) = \left[\lambda_{\sigma(t)_{-\ell}}  \hat g_{\sigma(t)_{-\ell}}^* \Gamma^T \frac{\bar s}{U_N} \right] \Gamma \hat g_{\sigma(t)_{-\ell}} 
$$
from which we see that 
\begin{equation}
\delta(t) \propto \Gamma \hat g_{[\sigma(t)]_{-\ell}}.
\label{eq:trend_prop}
\end{equation} 
Therefore, every specific switching action pattern that appears on the voltage profile is proportional to the eigenvector $\hat g_{[\sigma(t)]_{-\ell}}$, irrespective of other variables such as voltages $u$ and loads $s$ that describe the network operating state at the time. Thus, $g_{[\sigma(t)]_{-\ell}}$ can be seen as the \emph{particular signature} of the switch action. This fact is the cornerstone for the topology detection algorithm in this paper.
%

\section{Topology Detection Algorithm}
\label{sec:algorithm_statement}


Assuming the distribution network physical infrastructure and the initial switches status are known, we can construct a  \emph{library} $\mathcal L$ in which we collect all the normalized products between $(I - e_0 \1^T)$ and the eigenvectors for all possible switches action
\begin{equation}
\mathcal L_{\sigma(t-1)} = \{ g_{[\sigma(t)]_{-\ell}} : [\sigma(t)]_{-\ell} = [\sigma(t-1)]_{-\ell} \}
\label{eq:part_library}
\end{equation}
where
\begin{equation}
g_{[\sigma(t)]_{-\ell}} = \frac{\Gamma \hat g_{[\sigma(t)]_{-\ell}}}{\|\Gamma \hat g_{[\sigma(t)]_{-\ell}}\|}
\label{eq:wgvct_lib}
\end{equation}
The next step is comparing the trend vector $\delta(t,t-1)$ with the entries in the library to identify which switch changed its status.
The detection process is stated in Algorithm \ref{alg:DetAlg}.
\begin{small}
\begin{algorithm}
\caption{Topology Changes Detection}
	\begin{algorithmic}[1]
    \Require At each time $t$, $\sigma(t-1)$, \emph{min{\textunderscore}proj} = 0.98  
		\State $\sigma(t) \leftarrow \sigma(t-1)$ 
		
		\State each PMU at each node $j$ record voltage phasor measurements $y_j(t)$
		\State the algorithm builds the trend vector $\delta(t,t-1)$
		\State the algorithm projects $\delta(t,t-1)$ in the library $\mathcal L_{\mathcal P,\sigma(t)} $ obtaining the set of values $$\mathcal C = \left \{c_{[\sigma(t)]_{-\ell}} = \left \| \left \langle \frac{\delta}{\| \delta \|},g_{[\sigma(t)]_{-\ell}} \right \rangle \right \|, g_{[\sigma(t)]_{-\ell}} \in \mathcal L \right \};$$
		\If{$\max \mathcal C \geq \emph{min{\textunderscore}proj}$ } 
			\State $\sigma(t)  \leftarrow \arg \max \mathcal C$
		\EndIf
		\end{algorithmic}
\label{alg:DetAlg}
\end{algorithm}
\end{small}
The comparison is made by projecting the normalized actual trend vector $\frac{\delta(t,t-1)}{\| \delta (t,t-1) \|}$ onto the topology library $\mathcal L_{\sigma(t-1)}$. The projection is performed with the inner product, and it allows us to obtain the projection index for each vector in  $\mathcal L_{\sigma(t-1)}$
\begin{equation}
c_{[\sigma(t)]_{-\ell}} = \left \| \left \langle \frac{\delta}{\| \delta \|}, g_{[\sigma(t)]_{-\ell}} \right \rangle \right \|.
\label{eq:proj}
\end{equation}
If $c_{[\sigma(t)]_{-\ell}} \simeq 1$, it means that $\delta$ is spanned by $g_{[\sigma(t)]_{-\ell}}$ and then that the switch $\ell$ changed its status. Because of the approximation \eqref{eq:approximate_solution}, the projection will never be exactly one. Therefore, we will use a heuristic threshold, called \emph{min{\textunderscore}proj}, based on numerous simulations to select the right switch. If projection be greater than the threshold, the associated switch is selected. Based on simulations, the \emph{min{\textunderscore}proj} is setted to 0.98.
If there is no switching action, the trend vector will be zero as all the  $c_{[\sigma(t)]_{-\ell}}$, and the algorithm will not reveal any topology transition. 
Notice that the projection value is used to detect the change time too, differently of what proposed in \cite{cavraroISGT2015}, where instead we used the norm of a matrix built by measurements (the \emph{trend matrix}).
With a slight abuse of notation, we will say that the maximizer of $\mathcal C$ is the switches status $\sigma$ such that 
$[\sigma]_{-\ell} = [\sigma(t)]_{-\ell}$, $[\sigma]_{\ell} = 1$ if $[\sigma(t)]_{-\ell}=0$ or vice-versa $[\sigma]_{\ell} = 0$ if $[\sigma(t)]_{-\ell}=1$ and $c_{[\sigma(t)]_{-\ell}}$ its the maximum element in $\mathcal C$.
We tacitly assumed so far that all the buses are endowed with a PMU, but this is not a realistic scenario fora  distribution network. In presence of few measurements device the algorithm works the same way. The only difference is that we are allowed to take the few voltage measures
\begin{align}
y &= I_{\mathcal P} \green_{\sigma(t)} \frac{\bar s }{U_N} + \1 U_N
\label{eq:fewmeas}
\end{align}
where $I_{\mathcal P} \in [0,1]^{\nosensors \times \nonodes}$ is a matrix that select the entries of $u$ where a PMU is placed, and $\mathcal P$ is the set of nodes endowed with PMU. The trend vector will become
\begin{equation}
\delta(t_1,t_2) = y(t_1) - y(t_2)
\label{eq:trendvect_fewpmu}
\end{equation}
The elements of the library vector and their dimension change too. In fact one can easily show, using \eqref{eq:fewmeas} and retracing \eqref{eq:ux2} and \eqref{eq:PhiMiller} that \eqref{eq:wgvct_lib} becomes
\begin{equation}
g_{[\sigma(t)]_{-\ell}} = \frac{I_{\mathcal P}\Gamma\hat g_{[\sigma(t)]_{-\ell}}}{\|I_{\mathcal P}\Gamma \hat g_{[\sigma(t)]_{-\ell}}\|}
\label{eq:wgvct_few_lib}
\end{equation}
Of course, if we have only few PMUs, we have to tackle the \emph{observability problem}, i.e. we have to find a way to place the PMUs such that we are able to detect topology changes. Therefore, we have to minimize number of PMUs and maintain the system observability for topology detection.

\section{Measurements and Loads Uncertainty}
\label{sec:noisyscen}
So far, we considered the case in which the measurement devices were not affected by noise and loads were static.
In reality, there is some noise associated with PMUs. If we take \eqref{eq:PMUmeasurements} and \eqref{eq:deltap} into account, the trend vector becomes
\begin{align}
\delta(t_1,t_2) =   I_{\mathcal P}(&\green_{\sigma(t_1)} - \green_{\sigma(t_2)}) \frac{\bar s (t_2)}{U_N} + e_{t_1} - e_{t_2} + \notag \\
& + \frac{I_{\mathcal P} \green_{\sigma(t_1)}}{U_N} \sum_{t = t_2}^{t_1-1} n_p(t) - i n_q(t)
\label{eq:noisy_trendvect}
\end{align}
Therefore measurement noise and load dynamics yield non-zero values for the trend vector, even if there has not been any switching action. The projection index \eqref{eq:proj} may have values near unity, leading to wrong topology detection. 

When a switching action happens, branches of the network are changed and current flows change respectively, thus causing abrupt voltages variations. Therefore it helps to avoid topology detection errors caused by load uncertainty to consider a proper threshold \emph{min{\textunderscore}norm} for the trend vector norm. Moreover the additive noise can make the projection value of the trend vector onto the library considerably lower than one, even if a topology change occurred. This fact prompts us to use a threshold on the maximum projection value, \emph{min{\textunderscore}proj}, over which we consider if the trend vector change is due to a topology transition. To increase the accuracy of topology detection, the following steps are added to the algorithm. 
We assume the ideal case without load and measurement uncertainty with the $\ell$-th switch change its status at time $t_1$. Consider the trend vector 
$$\delta(t,t-\tau) = y(t) - y(t-\tau).$$
For $t<t_1$ and $t \geq t_1 + \tau$ the projections of the trend vector onto the library are all equal to zero, because
$$\delta(t,t-\tau) = y(t) - y(t-\tau) = 0$$
Instead for $t_1 \leq t < t_1 + \tau$, the trend vector is
$$\delta(t,t-\tau) = \Gamma \Phi_{[\sigma(t)]_{-\ell}} \Gamma^T \frac{\bar s }{U_N}$$
leading to a cluster of algorithm time instant of length $\tau$ (or $\frac \tau f$ seconds),   in which the maximum projection coefficient will be almost one.
A possible solution is thus to consider a trend vector built using not two consecutive measures, but considering measures separated by $\tau$ algorithm time istants 
$$\delta(t,t-\tau) = y(t) - y(t-\tau).$$
Assume that a topology change has happened at time $t$ when we have a cluster of algorithm time intervals of length $\tau$ ($\frac \tau f$ seconds).
The former observations lead to the Algorithm~\ref{alg:DetAlg_noisy} for topology detection with measurements noise and load variation.
\begin{small}
\begin{algorithm}
\caption{Topology Change Detection with Uncertainty}
    %
	\begin{algorithmic}[1]
  \Require At each time $t$, we are given the variables $\sigma(t-1)$, minimizer$(t-1)$, length{\textunderscore}cluster$(t-1)$  
	\State $\sigma(t) \leftarrow \sigma(t-1)$ 
		
	\State each PMU at each node $j$ record voltage phasor measurements $y_j(t)$
	\State the algorithm builds the trend vector $\delta(t,t-\tau)$
	\If{$\|\delta(t,t-\tau)\| <$ min{\textunderscore}norm} 
			\State $\delta(t,t-\tau)  \leftarrow 0$
			\State $\text{minimizer}(t) = 0$
			\State length{\textunderscore}cluster$(t)=0$
	\Else
			\State the algorithm projects $\delta(t,t-\tau)$ in the particular library $\mathcal L_{\mathcal P,\sigma(t)} $ obtaining the set of values $$\mathcal C = \left \{c_{[\sigma(t)]_{-\ell}} = \left \| \left \langle \frac{\delta}{\| \delta \|},g_{[\sigma(t)]_{-\ell}} \right \rangle \right \|, g_{[\sigma(t)]_{-\ell}} \in \mathcal L \right \};$$
			\If{$\max \mathcal C > $min{\textunderscore}proj } 
					\State $\text{minimizer}(t) = \arg \min \mathcal C$
					\If{$\text{minimizer}(t) = \text{minimizer}(t-1)$ } 
					\State length{\textunderscore}cluster$(t) \leftarrow$ length{\textunderscore}cluster$(t-1)$ + 1
							\If{length{\textunderscore}cluster$(t) = \tau$}
									\State $\sigma(t) \leftarrow \text{minimizer}(t) $
							\EndIf
					\Else 
					\State length{\textunderscore}cluster$(t) \leftarrow 1$
					\EndIf
			\EndIf
	\EndIf
	\end{algorithmic}
\label{alg:DetAlg_noisy}
\end{algorithm}
\end{small}
%

\section{Results, Discussions and Conclusions}
\label{sec:results}
We tested our algorithm for topology detection on the IEEE 33-bus distribution test feeder \cite{parasher2014load}, which is illustrated in the Figure \ref{fig:ieee33}.
In this testbed, there are five switches (namely $S_1$, $S_2$, $S_3$, $S_4$, $S_5$) that can be opened or closed, thus leading to the set of 32 possible topologies $\topology_1,\dots,\topology_{32}$. 
Because of the ratio between the number of buses and the number of switches, some very similar topologies can occur (for example the topology where only $S_1$ is closed and the one in which only $S_2$ is closed). 
In the IEEE33-bus test case, Assumption \ref{ass:sameR/X} about line impedances does not hold, making the test condition more realistic. Each bus of the network represent an aggregate of five houses, whose power demand is described by the statistical Gaussian model \eqref{eq:deltap}.  
\begin{figure}[]
\centering	
\includegraphics[width=0.25\textwidth]{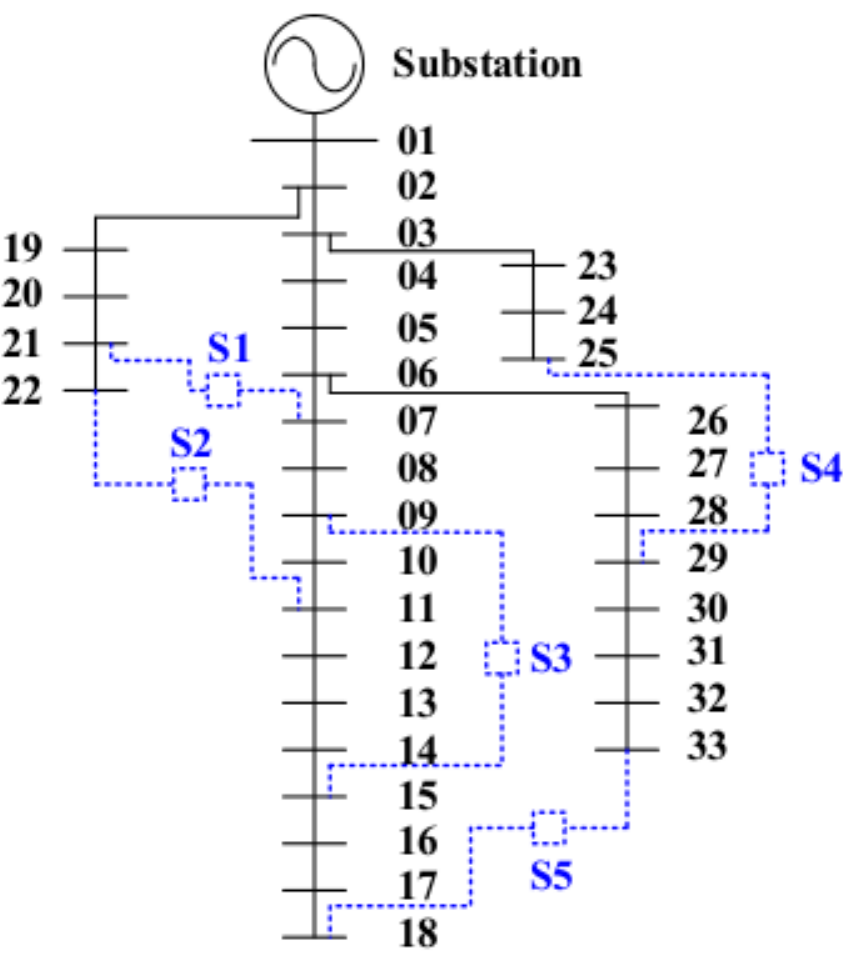}
\caption{Schematic representation of the IEEE33 buses distribution test case with the five switches}
\label{fig:ieee33}
\end{figure}
We tested the entire switch monitoring algorithm, in different situations. Firstly, we  consider the scenario in which the PMUs are affected by noise and the loads are not time varying, and then we add different levels of variation to them (associated with different measures frequencies). We assume that the buses are endowed with high precision devices, the $\mu$PMU \cite{microPMU}, affected by Gaussian noise such that $TVE \leq 0.05 \%$. It also complies with the IEEE standard C37.118.1-2011 for PMUs \cite{6111219}.
Furthermore, we vary the number and position of PMUs, considering the case in which every bus is endowed with a PMU, and the case in which we have only 7 PMUs deployed, whose placements have been chosen experimentally, after Monte Carlo simulations, as the one that minimizes the algorithm errors. Further research is needed to characterize a less onerous and more effective placement strategy.
The algorithm has been tested in each condition via 10000 Monte Carlo simulations The results are reported in Table~\ref{tab:sim33} and Table~\ref{tab:sim7}. We can see that, expected, the 33 PMUs scenario provides better performances. However the results with 7 PMUs are very close, showing the possibility of a satisfactory implementation of the algorithm also in a more realistic framework with few PMUs. Future developments include a deeper study about PMUs placement, better load characterization and further, analytic study of the thresholds \emph{min{\textunderscore}norm} and \emph{min{\textunderscore}proj} that yield the best performances of the algorithm.
\begin{table}
\caption{Results after 10000 runs with 33 PMUs}
\label{tab:sim33}
\centering
\begin{tabular}{lccccc}
SD [kV] & non & wrong  & decision & total & perc. of \\
	& detections & detection & errors & errors & errors (\%) \\
0  		& 0     & 50   	& 50    & 100	& 1.00 \\
0.184, ($f=1\text{ Hz}$)	& 0     & 64   	& 67    & 131 	& 1.31 \\
0.425, ($f=0.2\text{ Hz}$)	& 17		& 131   & 152 	& 300		& 3.00 \\
0.604, ($f=0.1\text{ Hz}$)	& 72   	& 211  	&	249		& 532		& 5.32 
\end{tabular}
\end{table}
\begin{table}
\caption{Results after 10000 runs with 7 PMUs}
\label{tab:sim7}
\centering
\begin{tabular}{lccccc}
Relative & non & wrong  & decision & total & perc. of \\
SD (\%)	& detections & detection & errors &errors & errors (\%) \\
0 		& 0     & 56   	&  56  & 112		& 1.12 \\
0.184, ($f=1\text{ Hz}$)	& 0     & 180   	&  185	& 365 	& 3.65 \\
0.425, ($f=0.2\text{ Hz}$)	& 31			& 199   	&  209  & 441		& 4.41 \\
0.604, ($f=0.1\text{ Hz}$)	& 76    & 245  		&	 298	& 619		& 6.19 
\end{tabular}
\end{table}
%
%

\bibliographystyle{IEEEtran}
\bibliography{bibTex_Topology}

\end{document}